\documentclass[a4paper, english]{scrartcl}
\usepackage{babel}

\usepackage[tt=false]{libertine}
\usepackage{microtype}
\usepackage{amsmath,amssymb,amsthm,mathtools,thmtools}
\usepackage{xparse}
\declaretheoremstyle[
  bodyfont=\normalfont,
  qed={\(\lrcorner\)}
]{corner-definition}

\declaretheorem{theorem}
\declaretheorem[sibling=theorem]{lemma}
\declaretheorem[sibling=theorem]{corollary}
\declaretheorem[sibling=theorem, style=remark]{remark}

\declaretheorem[style=corner-definition]{definition}

\usepackage{xcolor}
\definecolor{cite-green}{RGB}{0,115,0}
\definecolor{cau-main}{RGB}{155,10,125}
\usepackage[colorlinks, citecolor=cite-green, linkcolor=cau-main, pdfusetitle, pdfencoding=auto, psdextra]{hyperref}
\usepackage{cleveref}

\usepackage{todonotes}

\usepackage{shorthands}

\usepackage{enumitem}

\DeclarePairedDelimiterX\setImpl[1]{\{}{\}}{#1}
\NewDocumentCommand\set{sO{}m>{\TrimSpaces}o}{
  \IfValueTF{#4}{
    \setImpl[#2]{#3 : \IfBooleanTF{#1}{\text{#4}}{#4}}
  }{
    \setImpl[#2]{#3}
  }
}

\DeclarePairedDelimiterXPP\bo[1]{\cO}{(}{)}{}{#1}
\DeclarePairedDelimiter\enc{\langle}{\rangle}

\DeclareMathOperator\OPT{OPT}
\DeclareMathOperator\sol{sol}
\DeclareMathOperator\val{val}

\usepackage{xspace}

\DeclareDocumentCommand\problemfont{}{\scshape}
\DeclareDocumentCommand\DeclareProblem{mm}{\DeclareDocumentCommand{#1}{}{\textup{\problemfont #2}\xspace}}
\DeclareDocumentCommand\setc{soo}{\IfValueT{#2}{\IfValueT{#3}{(}#2\IfValueT{#3}{, #3)}-}{\problemfont\IfBooleanT{#1}{w}SetCover}\xspace}
\DeclareProblem\vertexc{VertexCover}
\DeclareProblem\domset{DominatingSet}

\DeclareDocumentCommand\complexityclassfont{}{\sffamily}
\DeclareDocumentCommand\DeclareComplexityClass{mm}{\DeclareDocumentCommand{#1}{}{\textup{\complexityclassfont #2}\xspace}}
\DeclareComplexityClass\fpt{FPT}
\DeclareComplexityClass\apx{APX}

\DeclareDocumentCommand\email{m}{\href{mailto:#1}{\nolinkurl{#1}}}

\usepackage{booktabs}

\usepackage[backend=biber,style=alphabetic,maxnames=10,maxalphanames=10]{biblatex}
\usepackage{csquotes}
\AtEveryBibitem{%
  \clearname{translator}%
  \clearname{editor}%
  \clearfield{pagetotal}%
  \clearfield{pages}%
  \clearfield{month}%
  \clearfield{date}%
  \iffieldundef{doi}{}{\clearfield{url}\clearfield{eprint}\clearfield{isbn}\clearfield{issn}}%
  \iffieldundef{eprint}{}{\clearfield{url}\clearfield{isbn}\clearfield{issn}}%
  \iffieldundef{isbn}{}{\clearfield{url}\clearfield{issn}}%
  \iffieldundef{issn}{}{\clearfield{url}}%
}

\usepackage{orcidlink}

\usetikzlibrary{backgrounds}
\usetikzlibrary{calc}
\newcommand{\convexpath}[2]{
[
    create hullnodes/.code={
        \global\edef\namelist{#1}
        \foreach [count=\counter] \nodename in \namelist {
            \global\edef\numberofnodes{\counter}
            \node at (\nodename) [draw=none,name=hullnode\counter] {};
        }
        \node at (hullnode\numberofnodes) [name=hullnode0,draw=none] {};
        \pgfmathtruncatemacro\lastnumber{\numberofnodes+1}
        \node at (hullnode1) [name=hullnode\lastnumber,draw=none] {};
    },
    create hullnodes
]
($(hullnode1)!#2!-90:(hullnode0)$)
\foreach [
    evaluate=\currentnode as \previousnode using \currentnode-1,
    evaluate=\currentnode as \nextnode using \currentnode+1
    ] \currentnode in {1,...,\numberofnodes} {
-- ($(hullnode\currentnode)!#2!-90:(hullnode\previousnode)$)
  let \p1 = ($(hullnode\currentnode)!#2!-90:(hullnode\previousnode) - (hullnode\currentnode)$),
    \n1 = {atan2(\y1,\x1)},
    \p2 = ($(hullnode\currentnode)!#2!90:(hullnode\nextnode) - (hullnode\currentnode)$),
    \n2 = {atan2(\y2,\x2)},
    \n{delta} = {-Mod(\n1-\n2,360)}
  in
    {arc [start angle=\n1, delta angle=\n{delta}, radius=#2]}
}
-- cycle
}

\title{Hardness of SetCover Reoptimization}
\author{Klaus Jansen\thanks{Funded by the Deutsche Forschungsgemeinschaft (DFG, German Research Foundation) -- Project number 528381760}\\Kiel University\\\email{kj@informatik.uni-kiel.de}, \orcidlink{0000-0001-8358-6796} \and Tobias Mömke\\University of Augsburg\\\email{moemke@informatik.uni-augsburg.de}, \orcidlink{0000-0002-2509-6972}\and Björn Schumacher\footnotemark[1]\\Kiel University\\\email{bsch@informatik.uni-kiel.de},  \orcidlink{0009-0002-2448-4794}}

\begin{document}

\maketitle

\section*{Abstract}
We study hardness of reoptimization of the fundamental and hard to approximate SetCover problem.
Reoptimization considers an instance together with a solution and a modified instance where the goal is to approximate the modified instance while utilizing the information gained by solution to the related instance.
We study four different types of reoptimization for (weighted) SetCover: adding a set, removing a set, adding an element to the universe, and removing an element from the universe.
A few of these cases are known to be easier to approximate than the classic SetCover problem.
We show that all the other cases are essentially as hard to approximate as SetCover.

The reoptimization problem of adding and removing an element in the unweighted is known to admit a PTAS.
For these settings we show that there is no EPTAS under common hardness assumptions via a novel combination of the classic way to show that a reoptimization problem is \NP-hard and the relation between EPTAS and \fpt.

\section{Introduction}
We consider the \setc Problem according to the following definition.
\begin{definition}[\setc*]\label{def:setc}\leavevmode
  Given a universe \(U\), a collection of sets \(\mathcal{S} \subseteq \cP(U)\) such that \(\bigcup\cS = U\), and a weight function \(w : \cS \to \bQ_{\geq 0}\), find a set \(S \subseteq \cS\) covering the universe \(U\) with minimum weight, i.e.\ that minimizes \(\sum_{s\in S}w(s)\).
  The unweighted version of the problem (\setc) is defined in the same way but \(w\) is the constant one-function, i.e.\ the objective is to find a cover of \(U\) with  as few sets from \(\cS\) as possible.
\end{definition}
The \setc problem is a fundamental problem and part of Karp's 21 \NP-complete problems.
A simple algorithm (greedily choosing the set that covers the most uncovered elements) achieves an approximation ratio of \(\bo{\ln\abs{U}}\) \cite{DBLP:journals/jcss/Johnson74a,DBLP:journals/dm/Lovasz75} and can be extended to the weighted case \cite{DBLP:journals/mor/Chvatal79}.
The exact approximation ratio for the greedy algorithm was analyzed in \cite{DBLP:journals/jal/Slavik97}.
There are several inapproximability results for \setc showing that the greedy algorithm is essentially optimal \cite{DBLP:journals/jacm/LundY94,DBLP:journals/jacm/Feige98,DBLP:conf/stoc/RazS97,DBLP:journals/talg/AlonMS06,DBLP:conf/stoc/DinurS14}.
The best known result that also has the weakest assumption is due to \textcite{DBLP:conf/stoc/DinurS14} and states that there is polynomial approximation algorithm for \setc with approximation ratio \((1 - \epsilon)\ln\abs{U}\) for any \(\epsilon > 0\) unless \(\PP = \NP\).
Regarding parameterized complexity, \setc is  known to be \(W[2]\)-complete \cite{DBLP:journals/siamcomp/DowneyF95} when parameterized by the number of sets in the solution but is fixed-parameter tractable when parameterized by \(\abs{U}\) or \(\abs{\cS}\).

As \setc is not just a hard problem but also hard to approximate it is interesting to consider what additional information helps to obtain better approximation ratios.
The extra information provided by reoptimization is a related instance and a solution with a certain quality.
The instance to solve is obtained by applying a specified modification to the related instance.
Typical modifications include by adding or removing certain parts of the input (vertices, edges, items) or changing weights or profits.

Reoptimization of (\textsc{w})\setc was considered before in \cite{DBLP:conf/waoa/BiloWZ08,DBLP:phd/basesearch/Zych12,mikhailyuk2010reoptimization}.
\textcite{DBLP:conf/waoa/BiloWZ08} considered the following reoptimization settings: adding a constant number of new elements to \(U\) and arbitrarily to sets in \(\cS\), removing a constant number of elements from sets in \(\cS\), complete removal of a constant number of elements, i.e.\ from \(U\) and all sets in \(\cS\), and adding an element from \(U\) to a set of \(\cS\).
They provide a general framework for reoptimization and show how it applies to these settings to obtain approximation algorithms.
For the settings which do not change \(U\) they provide approximation lower bounds based on edge addition and removal for \domset which utilizes to close connection of the two problems.
\textcite{mikhailyuk2010reoptimization} showed that there is an \((2 - \frac{1}{\ln\abs{\cS} + 1})\)-approximation algorithm for the reoptimization setting when up to \(\abs{\cS}\) elements from \(U\) are added to or removed from a set in \(\cS\).

We study four different types of reoptimization for \setc: adding a set (\(S^+\)), removing a set (\(S^{-}\)), adding an element to the universe (\(e^{+}\)), and removing an element from the universe (\(e^-\)).
The precise definitions of these modification are given in \Cref{sec:prelims}.
The settings \(e^+\) and \(e^-\) overlap with the settings studied in \cite{DBLP:conf/waoa/BiloWZ08,DBLP:phd/basesearch/Zych12} but there are no approximation lower bounds given for these settings.

\subsection{Related work}\label{sec:related-work}
In general reoptimization problem of \NP-hard problems is \NP-hard as described in \cite{DBLP:conf/sofsem/BockenhauerHMW08}.
The structure of the proof is roughly as follows:
For every instance, start from some trivial (polynomial solvable) instance and applying the local modification (and the reoptimization algorithm) until we obtain the desired instance (with a solution).
In the same way strongly \NP-hard problems remain strongly \NP-hard.
This rules out FPTAS in many cases and the best to hope for is an EPTAS.

It is a well-known fact that a reoptimization problem where the optimum changes only by a constant admit a PTAS in many cases due to a simple case distinction:
Either optimum is small and we can find a optimal solution (for example through enumeration) in polynomial time or the optimum is large and the given optimal solution can easily be turned into a solution for the modified instance which is good enough.
See, for example, Lemma 2 in \cite{DBLP:phd/basesearch/Zych12} or \cite{DBLP:conf/waoa/BiloWZ08} for a general description of this result.
If the considered problem is \fpt\ parameterized by the size of the solution the first case can be speed up to obtain an EPTAS.
For example, this is the case for the \apx-complete problem \vertexc.
Or the special case of \setc where the pairwise intersections of sets in \(\cS\) is bounded by a constant \(\Delta\) \cite{DBLP:journals/algorithmica/RamanS08} or the special case when the size of all sets is bounded by a constant.
Interestingly, the problem remains hard to approximate even in this case \cite{10.1145/380752.380839}, i.e.\ it is hard to approximate within a factor of \(\ln\Delta - \bo{\ln\ln\Delta}\) unless \(\PP=\NP\).

\subsection{Results}
A high-level overview of known and our results can be found in \Cref{fig:results}.
Hard meaning that the problem cannot be approximately significantly better than \setc (at most by a constant factor).
The exact results for \setc are given in \Cref{prop:setc:add-set:inapprox,prop:setc:rm-set:inapprox} and \Cref{prop:setc:adding:noeptas,prop:setc:removing:noeptas}.
And the results for \setc* are given in \Cref{prop:wsetc:add-set:inapprox,prop:wsetc:rm-set:inapprox}, \Cref{prop:wsetc:adding:inapprox}, and \Cref{prop:wsetc:removing:inapprox}.
\begin{table}[tbh]
  \caption{Overview of results.  Hard means effectively the same inapproximability results as \setc.  Gray results were known before.}
  \vspace{.25cm}
  \centering
  \begin{tabular}{rcccc}\toprule
           & \(S^{+}\) & \(S^{-}\) & \(e^{+}\)       & \(e^{-}\)     \\\midrule
    \setc  & hard      & hard      & \textcolor{gray}{PTAS}/no EPTAS   & \textcolor{gray}{PTAS}/no EPTAS \\
    \setc* & hard      & hard      & \textcolor{gray}{in \apx}/no PTAS & hard          \\\bottomrule
  \end{tabular}
  \label{fig:results}
\end{table}

Most of the hardness results follow the same pattern: For a given \setc instance we construct a reoptimization instance with a simple optimal solution such that the reoptimization algorithm has to effectively solve the original \setc instance.
The proofs to rule out an EPTAS are more involved.
Similar to the approach described above to show \NP-hardness of reoptimization problems, we start with an easy instance and use a presumed EPTAS at every step of the way to construct an \fpt algorithm parameterized by the size of the solution.
This similar to the classic result that an EPTAS for an optimization problem implies an \fpt algorithm parameterized by the size of the solution (cf.\ Theorem 1.32 in \cite{DBLP:series/txtcs/FlumG06}).
We show a general statement for this result in \Cref{prop:eptas-implies-fpt}.

\paragraph{Outline of the paper}
We begin by introducing notation and define our reoptimization settings in \Cref{sec:prelims}.
Then, we discuss our results for the unweighted and weighted case in \Cref{sec:setc,sec:wsetc} respectively.
We conclude by discussing out interesting questions for future work in \Cref{sec:conclusion}.

\section{Preliminaries}\label{sec:prelims}
For any \(n\in\bN_0\) let \([n] \coloneqq \set{1, 2, \ldots, n}\) and \([n]_0 \coloneqq [n]\cup\set{0}\).
With \(\enc{\cdot}\) we denote the encoding length.

For an instance \(I\in\cI\) of an optimization problem \(\cI\) we denote the set of valid solution by \(\sol(I)\) and the value of a solution \(S\in\sol(I)\) by \(\val(S)\).
The optimal value of a solution for an instance \(I\) is denoted by \(\OPT(I)\).

We define reoptimization problems in this paper as follows.
Let \(\cI\) be an optimization problem with a minimization objective.
The definition for a maximization objective is analogous.
Let \(\cM \subseteq \cI\times\cI\) be a relation on the set of instance (we write \(I\sim_{\cM}I'\) if \((I, I')\in\cM\)), which formalizes the valid modifications.
We write (\(\cM, \rho\))-\(\cI\) where \(\cI\) is the problem, \(\cM\) the modification, and \(\rho\) is the quality of the given solution.
An instance of this problem is a triple \((I, S, I')\) with \((I, I')\in\cM\), \(S\in\sol(I)\) and \(\val(S) \le \rho\OPT(I)\).
If an optimal solution is given we use \(\cM\)-\(\cI\) as a shorthand for (\(\cM, 1\))-\(\cI\).

Next, we precisely define the modifications we study for \setc*.
The definitions for \setc are the same, only with the restriction of the weight function to the constant one-function.
Let \((U, \cS, w)\) be a \setc* instance.
\begin{description}
\item[Adding a set (\(S^{+}\))] Given \(s'\in\cP(U)\setminus\cS\) and a value \(w_{s'}\in\bQ_{\geq 0}\) the modified instance is \((U,\cS'\coloneqq \cS\cup\set{s'},w')\) where
  \[
    w' : \cS \to \bQ_{\geq 0}, s \mapsto
    \begin{cases*}
      w_{s'} & if \(s = s'\)\\
      w(s)   & otherwise.
    \end{cases*}
  \]
\item[Removing a set (\(S^{-}\))] Given \(s'\in\cS\) the modified instance is \((U, \cS'\coloneqq \cS\setminus\set{s'}, w_{|\cS'})\).
\item[Adding an element (\(e^{+}\))] Given \(e\notin U\) and \(S_e\subseteq\cS\) the modified instance is \((U\cup\set{e}, \cS'\coloneqq (\cS\setminus S_e)\cup\set{s\cup\set{e}}[s\in S_e], w')\) where \(w' : \cS' \to \bQ_{\geq 0}, s \mapsto w(s \setminus \set{e})\).
\item[Removing an element (\(e^{-}\))] Given \(e\in U\) the modified instance is \(U\setminus\set{e}, \cS'\coloneqq \set{s\setminus\set{e}}[s\in\cS], w')\) where
  \[
    w' : \cS' \to \bQ_{\geq 0}, s \mapsto
    \begin{cases*}
      w(s) & if \(s\in\cS\)\\
      w(s\cup\set{e}) & otherwise.
    \end{cases*}
  \]
\end{description}
Note that since we require in \Cref{def:setc} that every instance of \setc is feasible that both the initial and the modified instance have to feasible, i.e.\ \(\bigcup \cS = U\).
This is not a restriction as we can test \(\bigcup \cS = U\) in polynomial time.

\section{Unweighted Set Cover}\label{sec:setc}
Changes that only change the optimum by a constant in the unweighted case include adding/removing a set with constant size and adding/removing a constant number of elements.
Thus, all these cases admit a PTAS as described in \Cref{sec:related-work}.
\subsection{\setc[\(S^{+}\)]}
We write our inapproximability results as a function of the universe such that they are compatible with the inapproximability results for \setc.
Typically, \(f\) will be slight variation of \(x \mapsto \ln x\), e.g.\ \(x\mapsto \frac{1}{3}\ln x\).
Note, that this notation still allows constant factor approximation by using a constant function.

To obtain a \setc instance with an obvious optimal solution we will duplicate every element of the universe and for every element of the universe add a set to cover it and its duplicate.
Thus, all the newly added sets form an optimal solution and adding a new set covering the duplicates exposes the original \setc instance.

We always build our reoptimization instances in such a way that they have obvious optimal solutions and thus showing our claims for \(\rho = 1\).
But since an optimal solution is also an \(\rho\)-approximate solution for any \(\rho > 1\) we show the claim for every \(\rho \geq 1\).
\begin{lemma}\label{prop:setc:add-set:approx}
  Let \(\rho \geq 1\).
  Let \(f : \bN \to \bQ_{\geq 1}\).
  An \(f(\abs{U})\)-approximation algorithm for \setc[\(S^{+}\)][\(\rho\)] implies an \(2f(2\abs{U})\)-approximation algorithm for \setc.
\end{lemma}
\begin{proof}
  Let \((U, \cS)\) be an \setc instance.
  W.o.l.g.\ assume \(\abs{U} \geq 1\) and \(\OPT((U, \cS)) < \abs{U}\), i.e.\ there exists a set in \(\cS\) with size at least 2.
  Obtain \(U'\) by adding a copy of every element such that \(\abs{U'} = 2\abs{U}\).
  For every \(u\in U\) let \(u'\) be the copy added to \(U'\).
  Next, define
  \[
    \cS' \coloneqq \cS \cup \set{\set{u, u'}}[u\in U].
  \]
  The optimal solution to this instance of \setc is \(S^{*} \coloneqq \cS' \setminus \cS = \set{\set{u, u'}}[u\in U]\).
  The construction is depicted in \Cref{fig:example:add-set:approx}.
  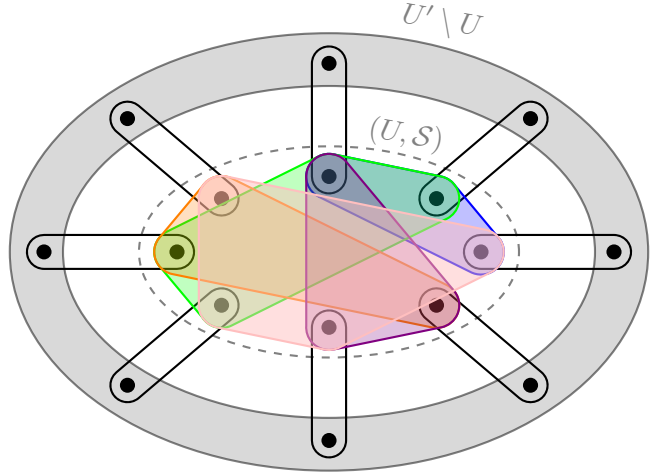
\begin{figure}[tbh]
    \centering
    \begin{tikzpicture}
      \foreach \r in {0,45,...,315} {
        \coordinate (c\r) at (\r:2 and 1);
        \coordinate (cc\r) at (\r:3.75 and 2.5);

        \draw[double, double distance=12pt, thick, line cap=round] (c\r) -- (cc\r);

        \node[fill, circle, inner sep=2pt] (\r) at (c\r) {};
        \node[fill, circle, inner sep=2pt] (\r') at (cc\r) {};
      }
      \draw[dashed, thick, opacity=.5] ellipse (2.5 and 1.4);
      \node[above, opacity=.5, rotate=-10] at (70:2.8 and 1.3) {\((U, \cS)\)};

      \fill[blue, opacity=.25] \convexpath{90,45,0}{.3cm};
      \draw[blue, thick] \convexpath{90,45,0}{.3cm};
      \fill[green, opacity=.25] \convexpath{225,180,90,45}{.3cm};
      \draw[green, thick] \convexpath{225,180,90,45}{.3cm};
      \fill[orange, opacity=.25] \convexpath{315,180,135}{.3cm};
      \draw[orange, thick] \convexpath{315,180,135}{.3cm};
      \fill[violet, opacity=.25] \convexpath{315,270,90}{.3cm};
      \draw[violet, thick] \convexpath{315,270,90}{.3cm};
      \fill[pink, opacity=.5] \convexpath{270,225,135,0}{.3cm};
      \draw[pink, thick] \convexpath{270,225,135,0}{.3cm};

      \fill[opacity=.15, even odd rule] ellipse (3.5 and 2.2) ellipse (4.2 and 2.9);
      \draw[thick, opacity=.5] ellipse (3.5 and 2.2);
      \draw[thick, opacity=.5] ellipse (4.2 and 2.9);

      \node[above, opacity=.5, rotate=-10] at (70:4.2 and 2.9) {\(U'\setminus U\)};
    \end{tikzpicture}
    \caption{Construction in \Cref{prop:setc:add-set:approx}.  The colorful center is the original instance, the outer dots are the duplicates, and the black sets are the sets that cover every element together with its duplicate.  The gray set is added as the local modification to cover all duplicates.}
    \label{fig:example:add-set:approx}
  \end{figure}

  Now, consider the \setc[\(S^{+}\)][\(\rho\)] instance with instance \((U', \cS')\), (optimal) solution \(S^{*}\), and modified instance \((U', \cS'\cup\set{U'\setminus U})\).
  Assume a solution of the modified instance contains a set of the form \(\set{u, u'}\) then there are effectively two cases:
  Either the solution only contains sets of this form or it contains the set \(U'\setminus U\).
  In both cases we can obtain a solution with smaller or equal value that does not contain any sets of the form \(\set{u, u'}\):
  For the first case notice that all elements from \(U\) can easily be covered by at most \(\abs{U} - 1\) sets and we additionally add the set \(U'\setminus U\).
  For the second notice that \(U'\setminus U\) is already in the solution meaning all elements from \(U'\) are already covered.
  Thus, we can simply exchange every set with the form \(\set{u, u'}\) with a set from \(\cS\) covering \(u\) while not increasing the size of the solution.
  Therefore, we may assume that no solution contains sets of the form \(\set{u, u'}\).
  A direct consequence is that \(\OPT((U', \cS'\cup\set{U'\setminus U})) = \OPT((U, \cS)) + 1\).

  We apply the presumed \(f(\abs{U})\)-approximation to the reoptimization instance and obtain a solution \(S\) such that \(\set{U'\setminus U} \subseteq S \subseteq \cS \cup \set{U'\setminus U}\) as discussed before.
  Notice that \(S' \coloneqq S \setminus \set{U'\setminus U}\) is a solution for the original instance.
  We have
  \begin{multline*}
    \abs{S'}
    = \abs{S} - 1
    \leq f(\abs{U'})\OPT((U', \cS'\cup\set{U'\setminus U})) - 1\\
    = f(2\abs{U})(\OPT((U, \cS)) + 1) - 1
    \leq 2f(2\abs{U})\OPT((U, \cS))
  \end{multline*}
  proving the claim.
\end{proof}
\begin{remark}
  The proof also shows that an exact algorithm for \setc[\(S^{+}\)][\(\rho\)] implies an exact algorithm for \setc.
  This can be seen by analyzing the last inequality chain when \(f = 1\).
\end{remark}

For every \(\alpha > 0\), \Cref{prop:setc:add-set:approx} shows that an \(\alpha\ln\abs{U}\) approximation for \setc[\(S^+\)][\(\rho\)] implies a \(2\alpha\ln(2\abs{U})\) approximation algorithm for \setc.
This approximation ratio is better than \((1 - 2\alpha)\ln\abs{U}\) for \(\alpha < \frac{1}{2}\) and \(\abs{U} > \exp(\frac{2\alpha\ln 2}{(1-2\alpha)^2})\).
Combining this with the \cite{DBLP:conf/stoc/DinurS14} result, we obtain the following corollary.
\begin{corollary}\label{prop:setc:add-set:inapprox}
  For any \(\alpha \in (0, \frac{1}{2})\) and \(\rho \geq 1\) there is no \(\alpha\ln\abs{U}\)-approximation algorithm for \setc[\(S^+\)][\(\rho\)] unless \(\PP = \NP\).
\end{corollary}

\subsection{\setc[\(S^{-}\)]}
\begin{lemma}\label{prop:setc:rm-set:approx}
  Let \(\rho \geq 1\).
  Let \(f : \bN \to \bQ_{\geq 1}\).
  An \(f(\abs{U})\)-approximation algorithm for \setc[\(S^{-}\)][\(\rho\)] implies an \(f(\abs{U})\)-approximation algorithm for \setc.
\end{lemma}
\begin{proof}
  Given any \setc instance.
  W.l.o.g.\ assume that there is no set covering the entire universe
  We add a new set containing the entire universe to the sets.
  The newly added set is an optimal solution.
  Thus, we have constructed an \setc[\(S^{-}\)][\(\rho\)] instance (removing the newly added set) which has to exactly solve the original instance.
  Therefore, we obtain an \(f(\abs{U})\)-approximation for \setc.
\end{proof}
This results implies that the same inapproximability result that hold for \setc also hold for \setc[\(S^{-}\)].
Thus, the known \(\bo{\ln \abs{U}}\)-approximation is the best we can hope for under the assumption \(\PP\ne\NP\).
Formally stated in the following corollary.
\begin{corollary}\label{prop:setc:rm-set:inapprox}
  For any \(\epsilon\) and \(\rho \geq 1\) there is no \((1-\epsilon)\ln\abs{U}\)-approximation algorithm for \setc[\(S^-\)][\(\rho\)] unless \(\PP = \NP\).
\end{corollary}

\subsection{\setc[\(e^{+}\)]}
The existence of a PTAS for \setc[\(e^{+}\)] was already discussed earlier.
Here we focus on showing that the existence of an EPTAS is unlikely.
First, we show the general framework and apply it afterwards to \setc[\(e^{+}\)] by giving an appropriate construction.

Similar to the way \NP-hardness is shown for reoptimization problems (as described in \Cref{sec:related-work}) we can rule out the existence of an EPTAS under certain conditions.
\Cref{prop:eptas-implies-fpt} shows how an EPTAS for the reoptimization problem can be used to build an built an \fpt algorithm for original optimization problem.
\begin{lemma}[EPTAS implies Parameterized Algorithm]\label{prop:eptas-implies-fpt}
  Let \(\cI\) be an optimization problem with integral solution values and \(\cM\subseteq \cI\times\cI\) a modification.
  Given an EPTAS for the reoptimization problem \(\cM\)-\(\cI\), computable functions \(f, f' : \bN \to \bN\), and for every \(I\in\cI\) we can find \(I_0, I_1, \ldots, I_n\in\cI\) (in \(\poly(\enc{I})\) time) such that for every \(k\in\bN\) we have
  \begin{itemize}
  \item if \(\OPT(I) \leq k\) then \(\OPT(I_i) \leq f(k)\) for all \(i\in[n]_0\),
  \item we can find an optimal solution to \(I_0\) or decide that \(\OPT(I_0) > f(k)\) in time \(f'(k)\poly(\enc{I})\),
  \item \(I_{i-1} \sim_{\cM} I_{i}\) for all \(i\in[n]\), and
  \item \(\OPT(I_n) \leq f(k) \implies \OPT(I) \leq k\)
  \end{itemize}
  then \(\cI\) is fixed-parameter tractable parameterized by solution size.
\end{lemma}
\begin{proof}
  Suppose we have an EPTAS for \(\cM\)-\(\cI\) with running time \(g(\frac{1}{\epsilon})\poly(\enc{I})\).
  We construct an algorithm for \(\cI\) that given an instance \(I\in\cI\) and a \(k\in\bN\) either computes an optimal solution or decides that \(\OPT(I) > k\).
  This implies that \(\cI\) is fixed-parameter tractable parameterized by solution size.

  Let \(k\in\bN\) and \(I\in\cI\) an arbitrary instance.
  First we calculate an optimal solution \(S_0\) to \(I_0\) or decide that \(\OPT(I_0) > f(k)\) in time \(f'(k)\poly(\enc{I})\).
  If \(\OPT(I_0) > f(k)\) we know that \(\OPT(I) > k\) and we are done.

  Let \(\epsilon \coloneqq \frac{1}{f(k) + 1}\).
  Let \(\cA_{\epsilon}(I, S, I')\) be the solution that the EPTAS calculates for the instance \(I'\) given accuracy \(\epsilon\).
  Let \(S_i \coloneqq \cA_{\epsilon}(I_{i - 1}, S_{i - 1}, I_i)\) for all \(i\in[n]\).
  We have
  \[
    \abs{\val(S_i) - \OPT(I_i)} \leq \epsilon\OPT(I_i) = \frac{\OPT(I_i)}{f(k) + 1} < 1
  \]
  for any \(i\in[n]\) assuming \(\OPT(I) \leq k\).
  Thus, the EPTAS calculates an optimal solution in every step.
  Therefore, we can decide \(I\), wether \(\OPT(I) \leq k\) (due to \(\OPT(I_n) \leq f(k) \implies \OPT(I) \leq k\)), in time \(\bo{f''(k)\poly(\enc{I})}\) where
  \[
    f'' : \bN \to \bN, k \mapsto f'(k) + g(f(k) + 1)
  \]
  which is a computable function.
\end{proof}
Since \Cref{prop:eptas-implies-fpt} shows the existence of a fixed-parameter tractable algorithm for the original problem when the reoptimization problem admits an EPTAS, it allows to conditionally rule out an EPTAS when the problem is \(W[t]\)-hard for some \(t\ge 1\) and the preconditions for \Cref{prop:eptas-implies-fpt} are fulfilled.
We state this formally in the following corollary.
\begin{corollary}\label{prop:w-hard-implies-no-eptas}
  Let \(t\in\bN_{\ge 1}\).  Given \(W[t]\)-hard problem, a reoptimization variant of the problem, and a construction that fulfill the preconditions of \Cref{prop:eptas-implies-fpt}, then there is no EPTAS for the reoptimization problem, unless \(W[t] = \fpt\).
\end{corollary}
\(W[t] \ne \fpt\) for any \(t\geq 1\) is classic assumption in parameterized complexity.
Note that \(W[t] = \fpt\) implies that the ETH fails (cf.\ Theorem 29.4.1 in \cite{DBLP:series/txcs/DowneyF13}).

Now we give the construction for \setc[\(e^{+}\)] to apply \Cref{prop:w-hard-implies-no-eptas}.
\begin{lemma}\label{prop:setc:adding:noeptas}
  There is no EPTAS for \setc[\(e^{+}\)], unless \(\fpt = W[2]\).
\end{lemma}
\begin{proof}
  Let \(I = (U, \cS)\) be an arbitrary \setc instance and let \(m\coloneqq \abs{\cS}\) and \(\cS = \set{s_1, \ldots, s_m}\).
  We take \(m + 1\) fresh elements \(e_1, \ldots, e_{m + 1}\notin U\) and define the new instance \(I' = (U', \cS')\).
  \[
    U' \coloneqq U \sqcup \set{e_1, \ldots, e_{m + 1}} \qquad \cS' \coloneqq \set{s_i \sqcup \set{e_i}}[i\in[m]] \sqcup \set{e_1, \ldots, e_{m + 1}}
  \]
  We have \(\OPT(I) + 1 = \OPT(I')\), as a solution has to always contain the set \(\set{e_1, \ldots, e_{m + 1}}\).
  Next consider the instance \(I_0 = (U'', \cS'')\) with
  \[
    U'' \coloneqq \set{e_1, \ldots, e_{m + 1}} \qquad \cS'' \coloneqq \set{\set{e_i}}[i\in[m]] \sqcup \set{e_1, \ldots, e_{m + 1}}
  \]
  for which \(\set{e_1, \ldots, e_{m + 1}}\) is the optimal solution.
  Now we define instances \(I_1, \ldots, I_{\abs{U}}\) by adding the elements of \(\abs{U}\) one by one.
  Due to the introduction of the elements \(e_i\) this construction guarantees that all sets are different at all points and we have \(I_{\abs{U}} = I'\).
  Further, we have \(\OPT(I_i) \leq \OPT(I) + 1\) for all \(i\in[\abs{U}]_0\).

  This sequence of instances fulfills the conditions of \Cref{prop:eptas-implies-fpt} and thus there is no EPTAS for \setc[\(e^{+}\)] unless \(\fpt = W[2]\) by \Cref{prop:w-hard-implies-no-eptas} as \setc is \(W[2]\)-complete.
\end{proof}
\subsection{\setc[\(e^{-}\)]}
Similar to the previous section on \setc[\(e^{+}\)] we only rule out an EPTAS for \setc[\(e^{-}\)] under common assumptions.
In this case we do not have a direct application of \Cref{prop:w-hard-implies-no-eptas} but the proof structure remains similar.
The challenge in this case is to find a bigger instance where optimum does not increase arbitrarily but still has an obvious optimal solution.
To cope with this issues we instead show that a different but closely related problem is fixed-parameter tractable when an EPTAS exists.
The concrete problem is \domset in unit disk graphs which has a PTAS \cite{DBLP:conf/waoa/NiebergH05} but is still \(W[1]\)-hard \cite{DBLP:conf/iwpec/Marx06}.
In the proof we use the fact that a \domset instance can always be viewed as a \setc instance and we use the PTAS to cope with the aforementioned issue.
We start with an approximate solution to a \domset instance and add elements to make this solution an optimal solution.
Next, we remove the added elements until we reach the original instance.

\begin{lemma}\label{prop:setc:removing:noeptas}
  There is no EPTAS for \setc[\(e^{-}\)], unless \(\fpt = W[1]\).
\end{lemma}
\begin{proof}
  Suppose we have an EPTAS for \setc[\(e^{-}\)].

  Let \((G = (V, E), k)\) be a parameterized instance of \domset in unit disk graphs.
  W.l.o.g.\ assume that there are no vertices \(u, v\in V\) with \(N(u) \cup \set{u} = N(v) \cup \set{v}\).
  We apply the PTAS \cite{DBLP:conf/waoa/NiebergH05} as a 2-approximation and obtain a solution \(S \subseteq V\).
  If \(\abs{S} > 2k\) we know that \(\OPT > k\) and there is no solution of size at most \(k\).
  Otherwise, we build a \setc Instance as follows.
  We set \(U \coloneqq V\) and \(\cS \coloneqq f[V]\) where
  \[
    f : V \to \cP(U), v \mapsto N(v) \cup \set{v}
  \]
  Note that \(f\) is an injective function and thus the solutions to the \domset and the constructed \setc instance are in a one-to-one relation that preserves the number of elements in a solution.
  Thus, it suffice to find a solution in the constructed instance with value at most \(k\) or decide that \(\OPT((U, \cS)) > k\).
  To do this we use the presumed EPTAS.

  First we add \(\abs{S}\) new elements to fix the solution \(S\) and make it the only optimal solution.
  For every \(v\in S\) we add an additional \(v'\) to the universe and add it to \(f(v)\).
  Call the resulting instance \(I'\).
  The sets corresponding to the elements from \(S\) are the optimal solution for \(I'\).
  Now we remove the added elements one-by-one until we have an solution for \((U, \cS)\) as in the proof of \Cref{prop:eptas-implies-fpt}.
  This works because the optimum in every step is bounded by \(2k\).

  Therefore, we showed that \domset in unit disk graphs is fixed-parameter tractable parameterized by solution size which implies \(\fpt = W[1]\) since the problem is \(W[1]\)-hard \cite{DBLP:conf/iwpec/Marx06}.
\end{proof}

\section{Weighted Set Cover}\label{sec:wsetc}
\subsection{\setc*[\(S^{+}\)]}
In the weighted case we can sharpen the results obtained in \Cref{prop:setc:add-set:approx} by removing the factor of two before \(f\) which allows to improve the inapproximability bound compared to \Cref{prop:setc:add-set:inapprox}.
\begin{lemma}
  Let \(\rho \geq 1\).
  Let \(f : \bN \to \bQ_{\geq 1}\).
  An \(f(\abs{U})\)-approximation algorithm for \setc*[\(S^{+}\)][\(\rho\)] implies an \(f(2\abs{U})\)-approximation algorithm for \setc.
\end{lemma}
\begin{proof}[Proof sketch]
  The proof is essentially the same as the proof of \Cref{prop:setc:add-set:approx} where all sets have weight 1 except for the set \(U'\setminus U\) which gets weight 0.
  Thus, the optimum of the original instance and reoptimization instance is the same yielding the sharper result.
\end{proof}

\begin{corollary}\label{prop:wsetc:add-set:inapprox}
  For any \(\epsilon \in (0, 1)\) and \(\rho \geq 1\) there is no \((1 - \epsilon)\ln\abs{U}\)-approximation algorithm for \setc*[\(S^+\)][\(\rho\)] unless \(\PP = \NP\).
\end{corollary}
\begin{proof}[Proof sketch]
  Same arguments as used for \Cref{prop:setc:add-set:inapprox}.
\end{proof}

\subsection{\setc*[\(S^{-}\)]}
We obtain the following result as a straightforward corollary from \Cref{prop:setc:rm-set:approx}.
\begin{corollary}
  Let \(\rho \geq 1\).
  Let \(f : \bN \to \bQ_{\geq 1}\).
  An \(f(\abs{U})\)-approximation algorithm for \setc*[\(S^{-}\)][\(\rho\)] implies an \(f(\abs{U})\)-approximation algorithm for \setc.
\end{corollary}
\begin{corollary}\label{prop:wsetc:rm-set:inapprox}
  For any \(\epsilon\) and \(\rho \geq 1\) there is no \((1-\epsilon)\ln\abs{U}\)-approximation algorithm for \setc*[\(S^-\)][\(\rho\)] unless \(\PP = \NP\).
\end{corollary}

\subsection{\setc*[\(e^{+}\)]}
\setc*[\(e^{+}\)][\(\rho\)] can be approximated in polynomial with a ratio of \(1 + \rho\).
This follows from Conclusion 3 in \cite{DBLP:phd/basesearch/Zych12} where they showed this result for the addition of a constant number of elements.

We show that an approximation algorithm for \setc*[\(e^{+}\)] with approximation ratio below \(\frac{3}{2}\) implies a constant factor approximation for \setc.
The main idea is to add an extra set that covers the entire universe and is expensive but still an optimal solution.
Then, we add an element to the instance but not to the expensive set.
We show that any solution taking the expensive set weighs at least \(1.5\) times the optimum and thus an approximation algorithm for \setc*[\(e^{+}\)] with ration smaller \(1.5\) cannot choose this set.
Furthermore, for the algorithm to achieve the approximation ratio of \(1.5\), it has to find an approximate solution to the original \setc instance within a ratio of \(2\).
\begin{lemma}\label{prop:wsetc:adding:inapprox}
  \setc*[\(e^{+}\)] cannot be approximated within a factor smaller than \(1.5\) in polynomial time, unless \(\PP=\NP\).
\end{lemma}
\begin{proof}
  Assume we have an approximation algorithm for \setc*[\(e^{+}\)] with approximation ratio below \(\frac{3}{2}\).

  Let \(I = (U, \cS)\) be an instance of \setc.
  W.l.o.g.\ assume that there are no singleton sets.
  If there is an element of the universe only covered by a singleton set we know that this singleton set is always in the solution and we can remove it and the corresponding element from the instance.
  Adding them back later to the instance and the singleton set to a solution only improves the approximation guarantee.
  If there are other singleton sets we may remove them because any solution containing singleton sets can easily be modified to only contain non-singleton sets and not increasing the size of the solution.

  Next we add a singleton set for all elements and obtain the instance \(I' = (U, \cS')\) such that \(\OPT(I') = \OPT(I)\) and \(\abs{\cS'} \geq 2\OPT(I)\) as there is an optimal solution with no singleton sets and \(\abs{U} \geq \OPT(I)\).
  This will be important for the analysis later on.

  We may assume that we know \(\OPT(I')\) because we do the following construction for each value in \([\abs{U}]\) and output the best valid solution we obtain.
  We only analyze the construction for the right guess yielding an upper bound on the output of the constructed algorithm.

  We build a \setc* instance \(I''\) as follows.
  For every \(s\in\cS'\) we introduce a new element \(e_s\) to the universe.
  Call the new universe \(U'' = U \cup\set{e_S}[S\in\cS']\).
  The sets are
  \[
    \cS'' \coloneqq \set{s\cup\set{e_s}}[s\in\cS'] \cup \set{G, R}
  \]
  where \(G \coloneqq U''\) and \(R\coloneqq \set{e_s}[s\in\cS']\) with weight function
  \[
    w : \cS'' \to \bQ_{\geq 0}, s \mapsto
    \begin{cases*}
      2\OPT((U, \cS')) & if \(s = G\) \\
      \OPT((U, \cS')) & if \(s = R\) \\
      1 & otherwise.
    \end{cases*}
  \]
  The construction is depicted in \Cref{fig:example:add-element:approx}.
  There are three candidates for an optimal solution of \(I'' = (U'', \cS'', w)\):
  First, just the set \(G\) with weight \(2\OPT((U, \cS'))\), then the set \(R\) and an optimal solution for the instance \((U, \cS')\) with weight \(2\OPT((U, \cS'))\), and all sets except for \(G\) and \(R\) with weight at least \(2\OPT((U, \cS'))\) due to the introduction of singleton sets above.
  Thus, \(G\) is an optimal solution.

  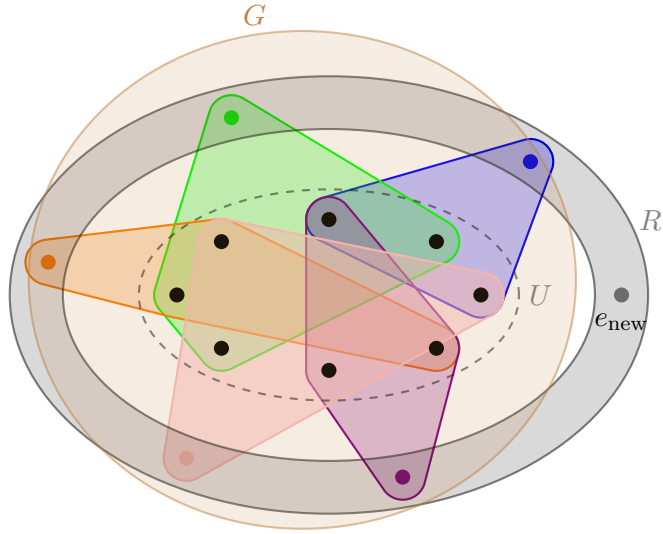
\begin{figure}[tbh]
    \centering
    \begin{tikzpicture}
      \foreach \r in {0,45,...,315} {
        \coordinate (c\r) at (\r:2 and 1);
        \coordinate (cc\r) at (\r:3.75 and 2.5);


        \node[fill, circle, inner sep=2pt] (\r) at (c\r) {};
      }
      \draw[dashed, thick, opacity=.5] ellipse (2.5 and 1.4);
      \node[right, opacity=.5] at (0:2.5 and 1.4) {\(U\)};

      \begin{scope}[on background layer]
        \node[fill, circle, inner sep=2pt, blue] (blue) at (45:3.75 and 2.5) {};
        \fill[blue, opacity=.25, on background layer] \convexpath{90,blue,0}{.3cm};
        \draw[blue, thick] \convexpath{90,blue,0}{.3cm};

        \node[fill, circle, inner sep=2pt, green] (green) at (110:3.75 and 2.5) {};
        \fill[green, opacity=.25, on background layer] \convexpath{225,180,green,45}{.3cm};
        \draw[green, thick] \convexpath{225,180,green,45}{.3cm};

        \node[fill, circle, inner sep=2pt, orange] (orange) at (170:3.75 and 2.5) {};
        \fill[orange, opacity=.25, on background layer] \convexpath{315,180,orange,135}{.3cm};
        \draw[orange, thick] \convexpath{315,180,orange,135}{.3cm};

        \node[fill, circle, inner sep=2pt, violet] (violet) at (285:3.75 and 2.5) {};
        \fill[violet, opacity=.25, on background layer] \convexpath{315,violet,270,90}{.3cm};
        \draw[violet, thick] \convexpath{315,violet,270,90}{.3cm};

        \node[fill, circle, inner sep=2pt, pink] (pink) at (240:3.75 and 2.5) {};
        \fill[pink, opacity=.5, on background layer] \convexpath{pink,135,0}{.3cm};
        \draw[pink, thick] \convexpath{pink,135,0}{.3cm};
      \end{scope}

      \fill[opacity=.15, even odd rule] ellipse (3.5 and 2.2) ellipse (4.2 and 2.9);
      \draw[thick, opacity=.5] ellipse (3.5 and 2.2);
      \draw[thick, opacity=.5] ellipse (4.2 and 2.9);
      \node[right, opacity=.5] at (20:4.2 and 2.9) {\(R\)};
      \node[opacity=.5, inner sep=2pt, fill, circle, label={-90:\(e_{\mathrm{new}}\)}] at (0:3.85 and 2.55) {};

      \fill[opacity=.15, brown] (-.35, .2) ellipse (3.6 and 3.3);
      \draw[thick, opacity=.5, brown] (-.35, .2) ellipse (3.6 and 3.3);
      \node[xshift=-.35cm, yshift=.2cm, brown, above] at (100:3.6 and 3.3) {\(G\)};
    \end{tikzpicture}
    \caption{Construction in \Cref{prop:wsetc:adding:inapprox}.  The black dots represent the original elements of the \setc instance.  The colorful sets are the sets of the original instance but each extended with a correspondingly colored element.  The extra sets \(G\) and \(R\) are the brown and gray sets respectfully.  The element that is added as the local modification is labeled \(e_{\mathrm{new}}\).}
    \label{fig:example:add-element:approx}
  \end{figure}

  The local modification is the addition of new element \(e_{\mathrm{new}}\) that is only added to the set \(R\).
  Call this instance \(I^{*}\).
  We apply our presumed algorithm on the instance \((I'', \set{G}, I^{*})\) and obtain a solution \(S\).
  Every solution for \(I^{*}\) has to include the modified version of \(R\) which means all elements \(e_s\) are covered.
  By combining an optimal solution for \((U, \cS')\) and the modified version of \(R\) we obtain an optimal solution with weight \(2\OPT((U, \cS'))\).
  Since every solution has to contain the set \(R\), a solution containing \(G\) has value at least \(3\OPT((U, \cS'))\).
  Therefore, the solution \(S\) does not contain the set \(G\).
  The number of sets from \(\cS'\) in the solution \(S\) is
  \begin{align*}
    \val(S) - w(R)
    &< \frac{3}{2}\OPT(I^{*}) - w(R)\\
    &= 3\OPT((U, \cS')) - \OPT((U, \cS'))\\
    &= 2\OPT((U, \cS')).
  \end{align*}
  Thus, we have constructed a 2-approximation for \setc.
\end{proof}
\begin{remark}
  This result can probably be sharpened for \setc*[\(e^{+}\)][\(\rho\)] for \(\rho > 1\) by adjusting the values for the sets \(R\) and \(G\).

  For \(\rho = 1\) it is possible to increase the bound when we have a constant \(c > 2\) such that \(\abs{\cS'} \geq c\OPT((U, \cS'))\) by adjusting the weights to \(w(G) = c\OPT((U, \cS'))\) and \(w(R) = (c - 1)\OPT((U, \cS'))\).
\end{remark}
\subsection{\setc*[\(e^{-}\)]}
\begin{lemma}\label{prop:wsetc:removing:approx}
  Let \(\rho \geq 1\).
  Let \(f : \bN \to \bQ_{\geq 1}\).
  An \(f(\abs{U})\)-approximation algorithm for \setc*[\(e^{-}\)][\(\rho\)] yields an \(f(\abs{U})\)-approximation algorithm for \setc*.
\end{lemma}
\begin{proof}
  Given an instance \((U, \cS, w)\) of \setc* we construct an \setc*[\(e^{-}\)] instance as follows.
  Let \(u'\notin U\) be a fresh element.
  The universe for the new instance is \(U' \coloneqq U \cup \set{u'}\), we keep the sets from \(\cS\) and single new set which contains all elements of \(U'\) with weight \(W \coloneqq (f(\abs{U}) + 1)\sum_{S\in\cS}w(S)\).
  The newly added set is an optimal solution as it contains all elements from \(U\) and is the only sets that contains \(u'\).
  Now using an \(f(\abs{U})\)-approximation algorithm on the constructed \setc*[\(e^{-}\)][\(\rho\)] instance solves the original instance with approximation factor \(f(\abs{U})\) since \(W > f(\abs{U})\OPT((U, \cS, w))\).
\end{proof}
\begin{corollary}\label{prop:wsetc:removing:inapprox}
  For any \(\epsilon\) and \(\rho \geq 1\) there is no \((1-\epsilon)\ln\abs{U}\)-approximation algorithm for \setc*[\(e^-\)][\(\rho\)] unless \(\PP = \NP\).
\end{corollary}
Note that we construct an algorithm in the proof of \Cref{prop:wsetc:removing:approx} that evaluates \(f(\abs{U})\) or at least needs to calculate an approximation of \(f(\abs{U})\).
If this is not possible in polynomial time we do not construct a polynomial time algorithm.
This is not a problem for \Cref{prop:wsetc:removing:inapprox} since we only need to calculate the natural logarithm.

\section{Conclusion}\label{sec:conclusion}
For all the reoptimization settings we discussed we essentially have matching approximation upper and lower bounds, except for \setc*[\(e^+\)][\(\rho\)] where have an upper bound of \(1 + \rho\) but only a lower bound of \(1.5\).
It is an interesting question what the best possible approximation ratio is for this problem.
Since it seems promising to improve in the construction in \Cref{prop:wsetc:adding:inapprox} it is probably not that close to \(1.5\).

Other perspectives to consider include bounded set sizes and bounded element frequency which have better approximation guarantees but still strong hardness results.
When the set sizes are bounded by a constant we obtain an EPTAS in cases we considered for \setc as discussed in \Cref{sec:related-work}.
It would be interesting to consider this restriction in the weighted case since most of our hardness construction use sets that are as big as or almost as big as \(U\).
For bounded set sizes as well as bounded frequencies the results in \cite{DBLP:conf/waoa/BiloWZ08,DBLP:phd/basesearch/Zych12} yield better approximation guarantees compared to the general case.

\printbibliography

\end{document}